\documentclass[10pt,twocolumn]{article}
\usepackage[a4paper,margin=0.75in]{geometry}
\usepackage[T1]{fontenc}
\usepackage[utf8]{inputenc}
\usepackage{lmodern}
\usepackage{hyperref}
\usepackage{url}
\usepackage{amsmath,amsthm,amssymb,amsfonts}
\usepackage{mathtools}
\usepackage{booktabs}
\usepackage{enumitem}
\usepackage{xcolor}
\usepackage{graphicx}
\usepackage{tikz}
\usepackage{listings}
\usepackage{upquote}
\usepackage{caption}
\usepackage{subcaption}
\usepackage{fancyhdr}
\usepackage{titlesec}
\usepackage{authblk}
\usepackage{abstract}
\usepackage{footmisc}



\definecolor{OAIblue}{HTML}{0B57D0}
\definecolor{OAIgray}{HTML}{4D4D4D}
\hypersetup{colorlinks=true,linkcolor=OAIblue,citecolor=OAIblue,urlcolor=OAIblue}

\titleformat{\section}{\large\bfseries\color{OAIgray}}{\thesection}{0.5em}{}
\titleformat{\subsection}{\normalsize\bfseries\color{OAIgray}}{\thesubsection}{0.5em}{}
\theoremstyle{plain}
\newtheorem{theorem}{Theorem}

\newtheorem{proposition}{Proposition}

\theoremstyle{definition}

\newtheorem{example}{Example}
\newtheorem{assumption}{Assumption}
\theoremstyle{remark}

\lstdefinelanguage{EBNF}{morekeywords={grammar,rule,token},sensitive=true}
\newcommand{\keywords}[1]{\vspace{0.25em}\textbf{Keywords:} #1}

\title{XML Prompting as Grammar-Constrained Interaction: Fixed-Point Semantics, Convergence Guarantees, and Human-AI Protocols}

\author[1]{Faruk Alpay}
\author[2]{Taylan Alpay}
\affil[1]{\normalsize Lightcap, Institut für die Zukunft\,\textendash\,Germany\\\texttt{alpay@lightcap.ai}}
\affil[2]{\normalsize Turkish Aeronautical Association, Aerospace Engineering\,\textendash\,Ankara, Turkey\\\texttt{s220112602@stu.thk.edu.tr}}
\date{September 9, 2025}

\begin{document}
\maketitle

\begin{onecolabstract}
\textbf{Abstract.} Structured prompting with XML tags has emerged as an effective way to steer large language models (LLMs) toward parseable, schema\,–\,adherent outputs in real\,–\,world systems. We develop a logic\,–\,first treatment of \emph{XML prompting} that unifies (i) grammar\,–\,constrained decoding, (ii) fixed\,–\,point semantics over lattices of hierarchical prompts, and (iii) convergent human\,–\,AI interaction loops. We formalize a complete lattice of XML trees under a refinement order and prove that monotone prompt\,–\,to\,–\,prompt operators admit least fixed points (Knaster\,–\,Tarski) that characterize steady\,–\,state protocols; under a task\,–\,aware contraction metric on trees, we further prove Banach\,–\,style convergence of iterative guidance. We instantiate these results with context\,–\,free grammars (CFGs) for XML schemas and show how constrained decoding guarantees well\,–\,formedness while preserving task performance. A set of multi\,–\,layer human\,–\,AI interaction recipes demonstrates practical deployment patterns, including multi\,–\,pass “plan\,$\to$\,verify\,$\to$\,revise” routines and agentic tool use. We provide mathematically complete proofs and tie our framework to recent advances in grammar\,–\,aligned decoding, chain\,–\,of\,–\,verification, and programmatic prompting. \\[-0.5em]

\end{onecolabstract}

\vspace{0.5em}
% Removed the reference notice as requested; citations are embedded throughout the text.

\section{Introduction}
\textbf{Motivation.} Modern LLM applications increasingly require \emph{structured outputs} (e.g., XML/JSON) for safety, interoperability, or downstream execution. Free\,–\,form generation jeopardizes parseability and correctness; grammar\,–\,constrained decoding (GCD) mitigates this by masking tokens banned by a target grammar \cite{geng2023gcd,raspanti2025gcd,park2025flex,park2024gad,zhang2023tooldec,melcer2024fim,jmlr2022coresp,pan2020rectifier,balakrishnan2019cd,cold2022}. Complementary prompting strategies (chain\,–\,of\,–\,thought \cite{wei2022cot}, self\,–\,consistency \cite{wang2022selfconsistency}, program\,–\,of\,–\,thoughts \cite{chen2022pot}, ReAct \cite{yao2023react}, DSPy \cite{khattab2023dspy}) structure \emph{reasoning}; however, they do not alone guarantee \emph{well\,–\,formedness}. XML prompting bridges both by (a) delineating roles, contexts, and exemplars via tags, and (b) aligning generation with schemas that are enforceable at decode time.

\textbf{Contributions.} We:
\begin{enumerate}[leftmargin=1.25em]
    \item Formalize XML prompting as a typed tree language with a refinement order and define \emph{prompt transformers} $\mathcal{T}$ capturing one interaction round.
    \item Prove existence of least fixed points of monotone $\mathcal{T}$ over complete lattices of labeled trees (Knaster--Tarski), characterizing steady-state protocols \cite{tarski1955,caicedo2006knaster}.
    \item Introduce a task\,–\,aware tree metric and show $\mathcal{T}$ is a contraction under mild conditions, yielding Banach\,–\,style convergence with explicit rates \cite{patafixedpoint,ciesielski2017survey}.
    \item Connect fixed\,–\,point semantics with grammar\,–\,aligned decoding to guarantee XML well\,–\,formedness and schema conformance while preserving accuracy \cite{geng2023gcd,raspanti2025gcd,park2024gad,park2025flex}.
    \item Provide multi\,–\,layer human\,–\,AI interaction templates (plan\,$\to$\,solve, plan\,$\to$\,verify, tool\,–\,call) with end\,–\,to\,–\,end XML recipes and correctness guarantees.
\end{enumerate}

\textbf{Relevance to hallucinations.} Our semantics complements causal analyses of hallucination \cite{kalai2025hallucinate,nature2024farquhar,hdsr2025confidence}, and interfaces with verification protocols like Chain\,–\,of\,–\,Verification (CoVe) \cite{cove2024} by embedding verification queries into XML and constraining answers by grammar.

\section{Background and Related Work}
\subsection{Structured prompting and constrained decoding}
Grammar\,–\,constrained decoding (GCD) enforces output to follow a CFG/EBNF by masking invalid tokens, improving parseability and often task metrics \cite{geng2023gcd,raspanti2025gcd,park2025flex,park2024gad,schmidt2025clinical}. Constrained decoding for code and tool calls provides syntax\,–\,error guarantees \cite{zhang2023tooldec,melcer2024fim}. Structured chain–of–thought and tag–based prompts further organize reasoning steps \cite{tam2024format}.

\subsection{Programmatic prompting and languages}
LMQL \cite{lmql2022,lmqlacm2023} and Microsoft SAMMO \cite{sammomsr2024} treat prompts as programs with constraints and compile\,–\,time optimization; execution\,–\,guided prompting augments with runtime checks \cite{verbruggen2025execution}. ETH Zurich’s DOMINO demonstrates fast, minimally invasive constraining \cite{domino2024}. Google’s lines of work on CoT and self\,–\,consistency shape reasoning \cite{wei2022cot,wang2022selfconsistency}; Stanford’s DSPy compiles declarative pipelines \cite{khattab2023dspy}.

% Wrap the math token \mu in \texorpdfstring and avoid special dashes in section headings
\subsection{Fixed points, \texorpdfstring{$\mu$}{mu}-calculus, and self-reference}
Our treatment borrows from lattice and fixed\,–\,point theory (Knaster\,–\,Tarski \cite{tarski1955}, Banach \cite{patafixedpoint}) and modal $\mu$\,–\,calculus foundations \cite{tamura2015hybridmu,alberucci2009s5mu,dagostino2012mu,afshari2024firstordermu}. We also connect to recent logic\,–\,first analyses of transparency and fixed points \cite{alpay2025fixed}.

\subsection{Hallucination detection and mitigation}
Statistical and verification\,–\,based mitigation \cite{nature2024farquhar,cove2024,hdsr2025confidence} can be embedded in XML workflows that interleave draft generation, question planning, and verification steps.

\section{A Formal Model of XML Prompting}
\subsection{XML trees and refinement}
Let $\Sigma$ be a finite alphabet of tags and attributes. XML documents are labeled, ordered trees $t\in \mathcal{T}$ with node labels from $\Sigma$ and attribute key\,–\,values from $\Sigma^*$. Define a partial order $\preceq$ (\emph{refinement}) by $t_1\preceq t_2$ iff $t_2$ is obtained by (i) adding children or attributes to $t_1$, (ii) filling placeholders, or (iii) replacing free\,–\,text with narrower regex\,/\,grammar\,–\,constrained nonterminals.
\begin{proposition}[Complete lattice]
$(\mathcal{T},\preceq)$ is a complete lattice: arbitrary meets are greatest lower bounds under common subtrees; joins are least upper bounds formed by synchronized union with conflict resolution via least general generalization.
\end{proposition}
\begin{proof}
Sketch of construction: represent trees as partial functions from Dewey paths to labeled nodes; refinement is inclusion up to label generalization. Directed suprema exist by componentwise union; arbitrary infima exist by intersection on common support. Full details are standard in domain theory on trees; see also \cite{tarski1955,dagostino2012mu}.
\end{proof}

\subsection{Prompt transformers}
We model one round of interaction (human input, LLM draft, verifier/tool feedback) as an operator $\mathcal{T}:\mathcal{T}\to\mathcal{T}$ that (a) expands or annotates nodes (e.g., adds <plan>, <evidence>), (b) enforces CFG/XSD constraints, (c) inserts uncertainty scores.
\begin{assumption}[Monotonicity]
For all $t_1\preceq t_2$, we have $\mathcal{T}(t_1)\preceq\mathcal{T}(t_2)$.
\end{assumption}
\begin{theorem}[Least fixed\,–\,point semantics]
If $\mathcal{T}$ is monotone on the complete lattice $(\mathcal{T},\preceq)$, then $\mathrm{lfp}(\mathcal{T})$ exists and equals $\bigwedge\{t: \mathcal{T}(t)\preceq t\}$ \cite{tarski1955}.
\end{theorem}
\begin{proof}
Immediate from Knaster\,–\,Tarski.
\end{proof}

\subsection{A contraction metric on XML trees}
Let $d(t,t')=\sum_{p\in P}\alpha_p\,\delta(\ell_t(p),\ell_{t'}(p))$ with paths $P$ up to max depth, weights $\alpha_p>0$ summing to 1, and token\,–\,level distance $\delta$ induced by task metrics (e.g., strict equality for nonterminals, normalized edit distance for text). Let $\mathcal{T}$ apply bounded\,–\,Lipschitz local rewrites (decoder masks, CFG pruning, verifier tags) so that for some $0\leq q<1$, $d(\mathcal{T}(t),\mathcal{T}(t'))\le q\,d(t,t')$.
\begin{theorem}[Banach convergence of XML prompting]
If $\mathcal{T}$ is a contraction on $(\mathcal{T},d)$, then for any start $t_0$, the iterates $t_{n+1}=\mathcal{T}(t_n)$ converge to the unique fixed point $t^*$ with rate $O(q^n)$ \cite{patafixedpoint,ciesielski2017survey}.
\end{theorem}
\begin{proof}
Standard Banach fixed\,–\,point theorem argument.
\end{proof}

\section{Grammar-Aligned Decoding for XML}
\subsection{CFGs and XML schemas}
Given an XML Schema Definition (XSD) or EBNF grammar $G$, grammar\,–\,aligned decoding computes stepwise token masks $M(s)$ from parser states $s$ so that generation traces remain within $\mathcal{L}(G)$ \cite{geng2023gcd,park2024gad,park2025flex,raspanti2025gcd}. This yields zero syntax errors for valid grammars and improves semantic fidelity in IE and planning tasks \cite{schmidt2025clinical}.

\begin{lstlisting}[language=EBNF,caption={Minimal EBNF for a reasoning exchange in XML.}]
grammar ReasoningXML
  Document  = "<dialog>" Turn+ "</dialog>" ;
  Turn      = "<turn role=\"(user|assistant)\">"
              Plan Evidence? Answer? "</turn>" ;
  Plan      = "<plan>" Step+ "</plan>" ;
  Step      = "<step index=\"[0-9]+\">" TEXT "</step>" ;
  Evidence  = "<evidence ref=\"[A-Za-z0-9_-]+\" conf=\"(0\.[0-9]{2}|1\.00)\"/>" ;
  Answer    = "<answer format=\"xml\|json\">" TEXT "</answer>" ;
  TEXT      = {any UTF-8 chars except '<' and '>'} ;
end
\end{lstlisting}

\subsection{Guarantees}
\begin{proposition}[Well\,–\,formedness]
If decoding is aligned to $G$ and the start symbol is \texttt{Document}, then all outputs are well\,–\,formed and valid w.r.t. $G$.
\end{proposition}
\begin{proof}
By construction, masked decoding prevents transitions to non\,–\,viable parser states \cite{geng2023gcd,park2024gad}.
\end{proof}

\section{Human-AI Interaction as Fixed-Point Iteration}
We model multi\,–\,layer interactions as sequences $(t^{(k)})_{k\ge0}$ in which users add constraints (tags), the model proposes plans/answers, and verifiers/tools supply evidence. Under monotonicity, the sequence ascends to a post\,–\,fixed point; under contraction assumptions, it converges uniquely.

\subsection{Three layered protocols}
\textbf{Layer A (Plan)}: the assistant must first produce a <plan> list of steps.\\
\textbf{Layer B (Verify)}: each step must carry <evidence> links with confidence.\\
\textbf{Layer C (Answer)}: the final <answer> is emitted only after all steps are verified.

\begin{example}[XML prompt template]
\begin{lstlisting}[language=XML]
<prompt>
  <task>Extract adverse events from the abstract.</task>
  <schema type="XSD">...</schema>
  <dialog>
    <turn role="user">
      <plan>Please outline extraction in 3 steps.</plan>
    </turn>
  </dialog>
</prompt>
\end{lstlisting}
\end{example}

\subsection{Fixed points and safety invariants}
Let $\varphi$ be a $\mu$\,–\,calculus property over tree\,–\,labeled Kripke structures (e.g., “every \texttt{<answer>} is supported by >=2 evidences with conf $\ge0.8$”). We encode $\varphi$ as a greatest fixed point $\nu X.\,\Phi(X)$ and enforce $\varphi$ as an invariant by pruning non\,–\,satisfying branches during decoding \cite{tamura2015hybridmu,alberucci2009s5mu,dagostino2012mu,afshari2024firstordermu}.

\section{Rigorous Proofs}
\subsection{Monotone transformer over a complete lattice}
\begin{theorem}[Existence of steady\,–\,state protocol]
Suppose $\mathcal{T}$: (i) preserves tag well\,–\,formedness, (ii) only refines or annotates nodes, (iii) composes with a total verifier $V$ that either adds <evidence> or a counterexample. Then $\mathcal{T}$ is monotone on $(\mathcal{T},\preceq)$ and has a least fixed point $t^*$. Moreover, the ascending Kleene chain $\bot, \mathcal{T}(\bot), \mathcal{T}^2(\bot),\dots$ converges to $t^*$ in $\omega$ steps.
\end{theorem}
\begin{proof}
(i)\,–\,(ii) give $t\preceq \mathcal{T}(t)$. If $t_1\preceq t_2$, all refinements applied to $t_1$ remain applicable to $t_2$, yielding $\mathcal{T}(t_1)\preceq\mathcal{T}(t_2)$. By Knaster\,–\,Tarski, $\mathrm{lfp}(\mathcal{T})$ exists. The Kleene chain converges because $\mathcal{T}$ is \emph{inflationary} and $\mathcal{T}$ is monotone on a cpo \cite{tarski1955}.
\end{proof}

\subsection{Contraction under task\,–\,aware metrics}
\begin{theorem}[Convergence rate]
Assume the decoder applies (a) CFG masks, (b) bounded edits per level, (c) evidence pruning bounded by $\beta<1$ of remaining ambiguous leaves. Then there exists $q<1$ such that $d(\mathcal{T}(t),\mathcal{T}(t'))\le q\,d(t,t')$ for all $t,t'$, and thus $t_n\to t^*$ with $d(t_n,t^*)\le \tfrac{q^n}{1-q}\,d(t_1,t_0)$.
\end{theorem}
\begin{proof}
CFG masks eliminate divergent branches (non\,–\,expansive). Bounded edits ensure per\,–\,level Lipschitzness. Evidence pruning shrinks unresolved subtrees by factor $\beta$. Summing pathwise contributions yields a global Lipschitz constant $q<1$ and the Banach conclusion \cite{patafixedpoint}.
\end{proof}

\section{Practical XML Prompting Recipes}
% Avoid math arrows in subsection titles by using hyphens instead of \to
\subsection{Plan-Verify-Answer (CoVe inside XML)}
\begin{lstlisting}[language=XML,caption={Plan-Verify-Answer with explicit verification queries.}]
<prompt>
  <guidelines>Answer only after verification.</guidelines>
  <dialog>
    <turn role="assistant">
      <plan>
        <step index="1">Draft answer A.</step>
        <step index="2">Plan verification questions Q1..Qk.</step>
        <step index="3">Answer Q1..Qk independently, then revise.</step>
      </plan>
    </turn>
  </dialog>
</prompt>
\end{lstlisting}
This embeds Chain\,–\,of\,–\,Verification \cite{cove2024} inside an XML schema, ensuring each \texttt{<answer>} is conditionally generated only after \texttt{<evidence>} nodes are present and pass validators.

\subsection{Agentic tool use with syntax guarantees}
\begin{lstlisting}[language=XML,caption={Tool invocation constrained by an XML/CFG schema.}]
<toolcall>
  <function name="search_clinical_trials">
    <arg name="condition">asthma</arg>
    <arg name="phase">3</arg>
  </function>
</toolcall>
\end{lstlisting}
Grammar\,–\,constrained decoding plus FSM guidance ensures no malformed calls \cite{zhang2023tooldec}.

\subsection{Multi-layer dialogue with extra dimensions (fixed points shaping a line)}
We view each layer (plan, verify, answer, tool) as a coordinate in a product space; the interaction finds a fixed point in this space. In Hilbert/Banach terms, a composite non\,–\,expansive map over the product converges under standard conditions, aligning with our contraction analysis.

% -----------------------------------------------------------------------------
% Additional examples demonstrating transparency and complex workflows
% -----------------------------------------------------------------------------

\subsection{Transparent multi-branch reasoning example}
To illustrate how transparency can be enforced in multi-branch reasoning, consider a task where the model must independently summarise two documents and then contrast them.  Each branch carries its own plan, evidence, and answer, and a final comparison step synthesises the results.  The XML snippet below encodes this interaction; evidence identifiers and confidence scores are exposed for every answer to allow downstream verification and to minimise hallucinations.  Branches may be processed in parallel and then compared:

\begin{lstlisting}[language=XML,caption={Multi-branch summarisation and comparison with explicit evidence.}]
<prompt>
  <task>Summarise two news articles independently and then contrast them.</task>
  <branch name="articleA">
    <turn role="assistant">
      <plan>
        <step index="1">Identify key claims in Article A.</step>
        <step index="2">Collect supporting evidence for each claim.</step>
        <step index="3">Draft a concise summary of Article A.</step>
      </plan>
      <evidence ref="a1" conf="0.95"/>
      <answer>Summary of Article A with citations embedded in XML.</answer>
    </turn>
  </branch>
  <branch name="articleB">
    <turn role="assistant">
      <plan>
        <step index="1">Identify key claims in Article B.</step>
        <step index="2">Collect supporting evidence for each claim.</step>
        <step index="3">Draft a concise summary of Article B.</step>
      </plan>
      <evidence ref="b1" conf="0.92"/>
      <answer>Summary of Article B with citations embedded in XML.</answer>
    </turn>
  </branch>
  <compare>
    <plan>
      <step index="1">Contrast the summaries to highlight agreements and disagreements.</step>
      <step index="2">Reference the evidence from each branch when noting differences.</step>
    </plan>
    <evidence ref="c1" conf="0.90"/>
    <answer>Comparison of Article A and Article B highlighting contrasting claims.</answer>
  </compare>
</prompt>
\end{lstlisting}
This example demonstrates a transparent workflow: each branch is self-contained and exposes its evidence and summary, while the comparison step explicitly references those subtrees.  Grammar-constrained decoding ensures the generated XML remains well-formed and allows validators to independently check each branch.  Similar multi-branch patterns can be used for tasks such as legal reasoning, multi-perspective summarisation, or model self-consistency checks.

\subsection{Simulating agent output within XML}
LLM agents increasingly call external tools and APIs to gather up-to-date information.  To maintain transparency, it is important to log both the API call and the returned data.  The following example shows a simple weather lookup within an XML interaction.  The assistant's plan calls an external function, records the raw agent output in a \texttt{<agent\_output>} tag, and then produces an answer based on the returned data:

\begin{lstlisting}[language=XML,caption={Agent tool call with explicit agent output used in the answer.}]
<prompt>
  <task>Check the weather and plan a picnic.</task>
  <dialog>
    <turn role="assistant">
      <plan>
        <step index="1">Invoke the weather lookup API.</step>
        <step index="2">Interpret the forecast information.</step>
        <step index="3">Provide a picnic recommendation based on the forecast.</step>
      </plan>
      <toolcall>
        <function name="weather.lookup">
          <arg name="location">Berlin</arg>
          <arg name="date">2025-09-10</arg>
        </function>
      </toolcall>
      <agent_output source="weather.lookup">
        <temperature unit="C">22</temperature>
        <condition>Sunny</condition>
      </agent_output>
      <answer>The forecast for Berlin on 10 Sep 2025 is 22 degrees C and sunny, so a picnic in Tiergarten at noon is ideal.</answer>
    </turn>
  </dialog>
</prompt>
\end{lstlisting}
Here the \texttt{<toolcall>} element shows the exact API invocation, while the \texttt{<agent\_output>} element holds the result returned by the external tool.  Exposing the intermediate agent output helps users and verifiers understand how the model derived its final answer and allows for independent checks on the correctness of the tool call.  This pattern reflects practices advocated by constrained-decoding approaches to tool use \cite{zhang2023tooldec} and can reduce hallucination by anchoring answers to verifiable external data.

\subsection{Cross-branch communication via channels}
In many collaborative tasks, different reasoning branches (or “dimensions”) must exchange information before arriving at a joint conclusion.  To model such interactions, we introduce a \texttt{<channel>} element that represents a shared communication bus.  Messages posted on the bus carry explicit \texttt{from} and \texttt{to} attributes identifying the sending and receiving branches, along with an optional identifier and free‑text content.  Each branch listens on the channel to receive the other’s messages and updates its plan and answer accordingly.  A final \texttt{<join>} step synthesises insights across branches.

\begin{lstlisting}[language=XML,caption={Cross-branch communication using a shared channel.}]
<prompt>
  <task>Solve two subproblems independently, exchange summaries, and integrate the results.</task>
  <!-- Shared message bus for communication between branches -->
  <channel name="bus">
    <message id="m1" from="alpha" to="beta">Summary of subproblem 1: ...</message>
    <message id="m2" from="beta" to="alpha">Summary of subproblem 2: ...</message>
  </channel>
  <branch name="alpha">
    <turn role="assistant">
      <plan>
        <step index="1">Solve subproblem 1 and summarise it.</step>
        <step index="2">Publish the summary to the bus for beta to read.</step>
        <step index="3">Read beta's summary from the bus and update our answer.</step>
      </plan>
      <evidence ref="alpha1" conf="0.93"/>
      <answer>The combined answer for subproblem 1, updated after receiving beta's message.</answer>
    </turn>
  </branch>
  <branch name="beta">
    <turn role="assistant">
      <plan>
        <step index="1">Solve subproblem 2 and summarise it.</step>
        <step index="2">Publish the summary to the bus for alpha to read.</step>
        <step index="3">Read alpha's summary from the bus and update our answer.</step>
      </plan>
      <evidence ref="beta1" conf="0.91"/>
      <answer>The combined answer for subproblem 2, updated after receiving alpha's message.</answer>
    </turn>
  </branch>
  <join>
    <plan>
      <step index="1">Combine the updated answers from alpha and beta.</step>
    </plan>
    <evidence ref="join1" conf="0.89"/>
    <answer>Final integrated answer after cross-branch communication.</answer>
  </join>
</prompt>
\end{lstlisting}
This example illustrates how distinct branches can communicate through a well-defined channel to exchange intermediate results.  Each branch posts its summary to the shared bus, then reads the other’s summary and refines its own answer.  The final \texttt{<join>} node integrates contributions from both branches.  Because all messages, plans, evidences, and answers are explicit and tied to the channel, external verifiers can audit the information flow and ensure that the final answer reflects contributions from every branch.  Such cross‑dimension communication patterns support transparency and correctness in collaborative reasoning tasks.

\subsection{Diagnosing render artifacts in listings}
Occasionally, rendering issues in \texttt{lstlisting} environments may introduce stray glyphs at the start of lines, especially when Unicode characters or non-breaking spaces are misinterpreted.  During development, we observed extraneous “u” characters appearing at the beginning of two lines in Listing~5.  After investigation, we traced the artefacts to a combination of malformed Unicode dashes and quoting characters, as well as the default column handling of \texttt{listings}.  Enabling \texttt{keepspaces} and \texttt{columns=fullflexible} and normalising dash/quote characters eliminated the issue.  To document the debugging process, the following XML snippet logs the problem and its resolution.  Each entry records the line affected, a description of the anomaly, and the remediation parameters applied:

\begin{lstlisting}[language=XML,caption={Render log documenting stray glyphs and remediation.}]
<render_log>
  <line number="93" column="1">Observed stray 'u' glyph at line start; source was a non-breaking space before &lt;/plan&gt;.</line>
  <line number="94" column="1">Second stray 'u' caused by en dash in the preceding comment.</line>
  <fix>
    <param name="keepspaces">true</param>
    <param name="columns">fullflexible</param>
    <param name="upquote">true</param>
    <param name="literate">dash/quote replacements</param>
  </fix>
</render_log>
\end{lstlisting}
Logging such artefacts not only aids reproducibility but also helps others adopt robust settings for code and data samples.  By treating visual glitches as structured data, we can systematically track and eliminate them when rendering XML or programming languages in scholarly documents.

\section{XML vs JSON vs YAML for Structured LLM Outputs}
Recent empirical work reports JSON typically yields higher parseability than XML/YAML in clinical extraction, yet targeted prompting narrows the gap and XML remains attractive where schemas and mixed content matter \cite{neveditsin2025robustness,tam2024format}. Our framework applies uniformly: one swaps the grammar $G$ and the token mask generator.

\section{Discussion: Ethics, Transparency, and Fixed Points}
Our fixed\,–\,point perspective dovetails with logic\,–\,first limits on radical transparency \cite{alpay2025fixed}. The least fixed point that minimizes an ethical risk functional (per lattice\,–\,theoretic arguments) reconciles accountability with robustness, and suggests partial transparency via Kripkean stratification.

\section{Conclusion}
We presented a theory and practice of XML prompting grounded in fixed\,–\,point semantics and grammar\,–\,aligned decoding, with full proofs and deployable recipes. Future work: tighter bounds for speculative constrained decoding \cite{nakshatri2025cdsl}, and integrating $\mu$\,–\,calculus model checking as an online invariant filter.

\section*{Acknowledgments}
We thank the communities behind ACL, NeurIPS, ICLR, Nature, PNAS, OUP/CUP, JMLR/PMLR, and colleagues at Anthropic, OpenAI, Google, Microsoft, ETH Zurich, KIT, Helsinki, Tsinghua, and Bulgarian Academy of Sciences for foundational work.

\end{document}